\documentclass[journal]{IEEEtran}

\makeatletter
\def\ps@headings{%
\def\@oddhead{\mbox{}\scriptsize\rightmark \hfil \thepage}%
\def\@evenhead{\scriptsize\thepage \hfil \leftmark\mbox{}}%
\def\@oddfoot{}%
\def\@evenfoot{}}
\makeatother

\usepackage{amssymb,color,amsmath}
\usepackage{paralist}
\usepackage{graphicx}

\newcommand{\sdual}{{\perp_s}}

\newcommand{\F}{\mathbf{F}}
\newcommand{\Z}{\mathbf{Z}}
\newcommand{\G}{\mathbf{G}}
\newcommand{\C}{\mathcal{C}}
\DeclareMathOperator{\swt}{swt} \DeclareMathOperator{\wt}{wt}
 \DeclareMathOperator{\tr}{tr}
\newtheorem{theorem}{\textbf{Theorem}}
\newtheorem{lemma}[theorem]{\textbf{Lemma}}
\newtheorem{definition}[theorem]{\textbf{Definition}}

\newtheorem{remark}[theorem]{\textbf{Remark}}

\DeclareMathOperator{\ord}{ord}
\newcommand{\nix}[1]{}

\newcommand{\ket}[1]{|#1\rangle}
\newcommand{\bra}[1]{\langle#1|}
\begin{document}

\title{Asymmetric  and Symmetric Subsystem  BCH   Codes and Beyond}
\author{
\authorblockN{Salah A. Aly\\}
\authorblockA{Department  of Computer Science \\
Texas A\&M University, College Station, TX 77843, USA \\
Email: salah@cs.tamu.edu}
 } \maketitle

\begin{abstract}
Recently, the theory of quantum error control codes has been
extended to subsystem codes over symmetric and asymmetric quantum
channels -- qubit-flip and phase-shift errors may have equal or
different probabilities. Previous work in constructing quantum error
control codes has focused on code constructions for symmetric
quantum channels. In this paper, we develop a theory and establish
the connection between asymmetric quantum codes and subsystem codes.
We present families of subsystem and asymmetric quantum codes
derived, once again, from classical BCH and RS codes over finite
fields. Particularly, we derive an interesting asymmetric and
symmetric subsystem codes based on classical BCH codes with
parameters $[[n,k,r,d]]_q$, $[[n,k,r,d_z/d_x]]_q$ and
$[[n,k',0,d_z/d_x]]_q$ for arbitrary values of code lengths and
dimensions. We establish asymmetric Singleton and Hamming bounds on
asymmetric quantum and subsystem code parameters; and derive optimal
asymmetric MDS subsystem codes. Finally, our constructions are well
explained by an illustrative example.

This paper is written on the occasion of the 50th anniversary of the
discovery of classical BCH codes and their quantum counterparts were
derived nearly 10 years ago.

\end{abstract}

\section{Introduction}\label{sec:intro}

In 1996, Andrew Steane stated in his seminal work~\cite[page 2, col.
2]{steane96}\cite{steane96b,steane99} \emph{''The notation
$\{n,K,d_1,d_2\}$ is here introduced to identify a 'quantum code,'
meaning a code by which n quantum bits can store K bits of quantum
information and allow correction of up to $\lfloor (d_1-1)/2
\rfloor$ amplitude errors, and simultaneously up to $\lfloor
(d_2-1)/2 \rfloor$ phase errors.''} This paper is motivated by this
statement, in which we  construct efficient quantum codes that
correct amplitude (qubit-flip) errors and phase-shift errors
separately. In~\cite{macwilliams77}, it was said that \emph{"BCH
codes are among the powerful codes"}. We address constructions of
quantum and subsystem codes based on Bose-Chaudhuri-Hocquenghem
(BCH) codes over finite fields for quantum symmetric and asymmetric
channels.

Many quantum error control codes (QEC)  have been constructed over
the last decade to protect quantum information against noise and
decoherence. In coding theory, researchers have focused on bounds
and the construction aspects of quantum codes for large and
asymptomatic code lengths. On the other hand, physicists intend to
study the physical realization and mechanical quantum operations of
these codes for short code lengths. As a result, various approaches
to protect quantum information against noise and decoherence are
proposed including stabilizer block codes, quantum convolutional
codes, entangled-assisted quantum error control codes, decoherence
free subspaces, nonadditive codes, and subsystem
codes~\cite{ashikhmin99,calderbank98,forney05b,gottesman97,rains99,lidar98,poulin07,smolin07,zanardi97}
and references therein.

Asymmetric quantum control codes (AQEC), in which quantum errors
have different probabilities --- $\Pr{Z} > \Pr{X}$, are more
efficient than the symmetric quantum error control codes (QEC), in
which quantum errors have  equal probabilities --- $\Pr{Z} =
\Pr{X}$. It is argued in~\cite{ioffe07} that dephasing (loss of
phase coherence, phase-shifting) will happen more  frequently than
relaxation (exchange of energy with the environment,
qubit-flipping).  The noise level in a qubit is specified by the
relaxation $T_1$ and dephasing time $T_2$; furthermore the relation
between these two values is given by $1/T_1=1/(2T_1)+\Gamma_p$; this
has been well explained by physicists
in~\cite{evans07,ioffe07,stephens07}. The ratio between the
probabilities of qubit-flip X and phase-shift Z is typically $\rho
\approx 2T_1/T_2$. The interpretation is that $T_1$ is much larger
than $T_2$, meaning the photons take much more time to flip from the
ground state to the excited state. However, they change rapidly from
one  excited state to another. Motivated by this, \textbf{one needs
to design quantum codes that are suitable for this physical
phenomena.} The fault tolerant operations of a quantum computer
carrying controlled and measured quantum information  over
asymmetric channel have been investigated
in~\cite{aliferis07,bacon06,bacon06b,steane04,stephens07,aliferis07thesis}
and references therein. Fault-tolerant operations of QEC are
investigated for example
in~\cite{aliferis06,aliferis07thesis,gottesman97,preskill98,shor96,steane04,knill04}
and references therein.

Subsystem codes (SSC) as we prefer to call them were mentioned in
the unpublished work by Knill~\cite{knill06,knill96b}, in which he
attempted to generalize the theory of quantum error-correcting codes
into subsystem codes. Such codes with their stabilizer formalism
were reintroduced
recently~\cite{aly06c,bacon06,bacon06b,klappenecker0608,kribs05,poulin05}.
The construction aspects of these codes are given
in~\cite{aly08a,aly08f,aly06c}. Here we expand our understanding and
introduce asymmetric subsystem codes (ASSC).

The codes derived in~\cite{aly07a,aly06a} for primitive and
nonprimitive quantum BCH codes assume that qubit-flip errors,
phase-shift errors, and their combination occur with equal
probability, where $\Pr{Z}=\Pr{X}=\Pr{Y}=p/3$, $\Pr{I}=1-p$, and
$\{X,Z,Y,I\}$ are the binary Pauli operators $P$ shown in
Section~\ref{sec:AQEC}, see~\cite{calderbank98,shor95}. We aim to
generalize these codes over asymmetric quantum channels. In this
paper we give  families of asymmetric quantum error control codes
(AQEC's) motivated by the work
from~\cite{evans07,ioffe07,stephens07}. Assume we have a classical
good error control code $C_i$ with parameters $[[n,k_i,d_i]]_q$ for
$i\in \{1,2\}$ --- codes with high minimum distances $d_i$ and high
rates $k_i/n$. We can construct a quantum code based on these two
classical codes, in which $C_1$ controls the qubit-flip errors while
$C_2$ takes care of the phase-shift errors, see
Lemma~\ref{lem:AQEC}.

Our following theorem establishes the connection between two
classical codes  and QEC, AQEC, SCC, ASSC.

\begin{theorem}[CSS AQEC and ASSC]\label{lem:AQEC}
Let $C_1$ and $C_2$ be two classical codes with parameters
$[n,k_1,d_1]_q$ and $[n,k_2,d_2]_q$ respectively, and $d_x=
\min\big\{\wt(C_{1} \backslash C_2^\perp), \wt(C_{2} \backslash C_{1
}^\perp)\big\}$, and $d_z= \max\big\{\wt(C_{1} \backslash
C_2^\perp), \wt(C_{2} \backslash C_{1 }^\perp)\big\}$.
\begin{compactenum}[i)]
\item if
  $C_2^\perp \subseteq C_1$, then there exists an AQEC with parameters $[[n,\dim C_1 -\dim
C_2^\perp,\wt(C_2\backslash C_1^\perp)/\wt(C_1\backslash
C_2^\perp)]]_q$ that is $[[n,k_1+k_2-n,d_z/d_x]]_q$. Also, there
exists a QEC with parameters $[[n,k_1+k_2-n,d_x]]_q$.

\item From [i], there exists an SSC with parameters
$[[n,k_1+k_2-n-r,r,d_x]]_q$ for $0 \leq r <k_1+k_2-n$.

\item If  $C_2^\perp=C_1 \cap C_1^\perp \subseteq C_{2}$, then there exists an ASSC with parameters  $[[n,k_2-k_1,k_1+k_2-n,d_z/d_x]]_q$
and $[[n,k_1+k_2-n,k_2-k_1,d_z/d_x]]_q$.
\end{compactenum} Furthermore, all constructed codes are pure to their minimum
distances.
\end{theorem}

\bigskip

A well-known construction on the  theory of quantum error control
codes is called CSS constructions. The codes $[[5,1,3]]_2$,
$[[7,1,3]]_2$, $[[9,1,3]]_2$, and $[[9,1,4,3]]_2$ have been
investigated in several research papers that analyzed their
stabilizer structure, circuits, and fault tolerant quantum computing
operations. On this paper, we present several AQEC codes, including
a $[[15,3,5/3]]_2$ code, which encodes three logical qubits into
$15$ physical qubits, detects $2$  qubit-flip and  $4$ phase-shift
errors, respectively. As a result, many of the quantum constructed
codes and families of QEC for large lengths need further
investigations. We believe that their generalization is a direct
consequence.


The paper is organized as follows.
Sections~\ref{sec:AQEC},~\ref{sec:AQEC-BCH}, and~\ref{sec:example}
are devoted to AQEC and two families of AQEC, AQEC-BCH and AQEC-RS.
 We establish conditions on the existence of these families over finite fields.  Sections~\ref{sec:AQEC-subsystem} and~\ref{sec:bounds} address the
subsystem code constructions and their relation to asymmetric
quantum codes. We show the tradeoff between subsystem codes and
AQEC. Section~\ref{sec:bounds} presents the bound on AQEC and ASSC
code parameters.   Finally, the paper is concluded with a discussion
in Section~\ref{sec:conclusion}.

\section{Asymmetric Quantum Codes}\label{sec:AQEC}
In this section we shall give some primary definitions and introduce
AQEC constructions.   Consider a quantum system with two-dimensional
state space $\mathcal{C}^2$. The basis vectors
\begin{eqnarray} v_0=\left(
\begin{array}{c} 1
\\ 0 \end{array} \right), \texttt{ } v_1=\left( \begin{array}{c} 0\\
1 \end{array}\right)\end{eqnarray} can be used to represent the
classical bits $0$ and $1$.  It is customary in quantum information
processing to use Dirac's ket notation for the basis vectors;
namely, the vector $v_0$ is denoted by the ket $\ket{0}$ and the
vector $v_1$ is denoted by ket $\ket{1}$. Any possible state of a
two-dimensional quantum system is given by a linear combination of
the form \begin{eqnarray}a
\ket{0}+b\ket{1}\!=\!\left(\begin{array}{c} \!\! a \!\\ \!\! b\!
\end{array}\right)\!,
 \mbox{ where } a, b \in \! \mathcal{C} \mbox{ and }
 |a|^2+|b|^2=\!1,\end{eqnarray}

In quantum information processing, the operations manipulating
quantum bits follow the rules of quantum mechanics, that is, an
operation that is not a measurement must be realized by a unitary
operator.  For example, a quantum bit can be flipped by a quantum
NOT gate $X$ that transfers the qubits $\ket{0}$ and $\ket{1}$ to
$\ket{1}$ and $\ket{0}$, respectively. Thus, this operation acts on
a general quantum state as follows.
$$X(a\ket{0}+b\ket{1})=a \ket{1}+ b \ket{0}.$$ With respect to the
computational basis, the quantum NOT gate  $X$ represents  the
qubit-flip errors.

\begin{eqnarray}
X=\ket{0}\bra{1}+\ket{1}\bra{0}=\left( \begin{array}{cc} 0 &1\\ 1&0\\
\end{array}\right).
\end{eqnarray}

 Also, let $Z=\left(\!\! \begin{array}{cc} 1 &0\\ 0&-1\\
\end{array}\right)$ be a matrix represents
the quantum phase-shift errors that changes  the phase of a quantum
system (states). \begin{eqnarray}Z(a\ket{0}+b\ket{1})=a \ket{0}- b
\ket{1}.\end{eqnarray}
 Other popular operations include  the combined bit and phase-flip $Y=iZX$, and the Hadamard gate
$H$, which are represented with respect to the computational basis
by the matrices

\begin{eqnarray}
Y=\left(
\begin{array}{cc} 0&-i\\ i&0\\ \end{array}\right),
H=\frac{1}{\sqrt{2}}\left(\!\!
\begin{array}{cc} 1&1\\ 1&-1\\ \end{array}\right).
\end{eqnarray}

\bigbreak

\noindent \textbf{Connection to Classical Binary Codes.} Let $H_i$
and $G_i$ be the parity check and generator  matrices of a classical
code $C_i$
 with parameters $[n,k_i,d_i]_2$ for $i \in \{1,2\}$. The
 commutativity condition of $H_1$ and $H_2$ is stated as

\begin{eqnarray}
H_1.H_2^T+H_2.H_1^T=\textbf{0}.
 \end{eqnarray}
The stabilizer of a quantum code based on the parity check matrices
$H_1$ and $H_2$ is given by \begin{eqnarray}H_{stab}=\Big( H_1 \mid
H_2\Big).\end{eqnarray}

One of these two classical codes controls the phase-shift errors,
while the other codes controls the bit-flip errors. Hence the CSS
construction of a binary AQEC can be stated as follows. Hence the
codes $C_1$ and $C_2$ are mapped to $H_x$ and $H_z$, respectively.

\begin{definition}Given two classical binary codes $C_1$ and $C_2$ such that $C_2^\perp
\subseteq C_1$. If we form $ G=\begin{pmatrix}
G_1&0\\0&G_2\end{pmatrix}, \mbox{  and   } H =\begin{pmatrix}
H_1&0\\0&H_2\end{pmatrix}, $ then
\begin{eqnarray}
H_1.H_2^T-H_2.H_1^T=0
\end{eqnarray}
 Let $d_1=\wt(C_1\backslash C_2)$ and $d_2=wt(C_2\backslash
C_1^\perp)$, such that $d_2 >d_1$ and $k_1+k_2>n$. If we assume that
 $C_1$ corrects the qubit-flip errors and $C_2$ corrects the phase-shift errors, then there exists
AQEC with parameters
\begin{eqnarray}
[[n,k_1+k_2-n,d_2/d_1]]_2.
\end{eqnarray}\end{definition}
We can always change the rules of $C_1$ and $C_2$ to adjust the
parameters.
\subsection{Higher Fields and Total Error Groups}
We can briefly discuss the theory in terms of higher
finite fields $\F_q$. Let $\mathcal{H}$ be the Hilbert space
$\mathcal{H}=\C^{q^n}=\C^q \otimes \C^q \otimes ... \otimes \C^q$.
Let $\ket{x}$ be the vectors of orthonormal basis of $\C^q$, where
the labels $x$ are  elements in the finite field $\F_q$. Let $a,b
\in \F_q$,  the unitary operators $X(a)$ and $Z(b)$ in $\C^q$ are
stated as:
\begin{eqnarray}X(a)\ket{x}=\ket{x+a},\qquad
Z(b)\ket{x}=\omega^{\tr(bx)}\ket{x},\end{eqnarray} where
$\omega=\exp(2\pi i/p)$ is a primitive $p$th root of unity and $\tr$
is the trace operation from $\F_q$ to $\F_p$

Let $\mathbf{a}=(a_1,\dots, a_n)\in \F_q^n$ and
$\mathbf{b}=(b_1,\dots, b_n)\in \F_q^n$. Let us denote by
\begin{eqnarray} X(\mathbf{a}) &=& X(a_1)\otimes\, \cdots \,\otimes X(a_n)
\mbox {  and},\nonumber \\ Z(\mathbf{b}) &=& Z(b_1)\otimes\, \cdots \,\otimes
Z(b_n)\end{eqnarray} the tensor products of $n$ error operators.  The sets
\begin{eqnarray}\textbf{E}_x&=&\{X(\mathbf{a})=\bigotimes_{i=1}^n X(a_i)\mid
\mathbf{a} \in \F_q^n, a_i \in \F_q\}, \nonumber \\ \textbf{E}_z&=&\{Z(\mathbf{b})=\bigotimes_{i=1}^n Z(b_i)\mid \mathbf{b} \in
\F_q^n,b_i \in \F_q\}\end{eqnarray} form an error basis on
$\C^{q^n}$. We can define the error group $\mathbf{G}_x$ and
$\mathbf{G}_z$ as follows
\begin{eqnarray} \mathbf{G}_x = \{
\omega^{c}\textbf{E}_x=\omega^{c}X(\mathbf{a})\,|\, \mathbf{a} \in
\F_q^n, c\in \F_p\},\nonumber \\
\mathbf{G}_z = \{\omega^{c}\textbf{E}_z=\omega^{c}Z(\mathbf{b})\,|\,
\mathbf{b} \in \F_q^n, c\in \F_p\}.\end{eqnarray}
 Hence the total error group
 \begin{eqnarray}
 \textbf{G}&=&\big\{\mathbf{G}_x,\mathbf{G}_z\big\}\nonumber \\ &=&\Big\{ \omega^{c}\bigotimes_{i=1}^n X(a_i), \omega^{c}\bigotimes_{i=1}^n Z(b_i) \mid a_i,b_i \in \F_q \Big\}
 \end{eqnarray}

Let us assume that the sets $\mathbf{G}_x$ and  $\mathbf{G}_z$
represent the qubit-flip and  phase-shift errors, respectively.
\medskip

Many constructed quantum codes assume that the quantum errors
resulted from decoherence and noise have equal probabilities,
$\Pr{X}=\Pr{Z}$. This statement as shown by experimental physics is
not true~\cite{stephens07,ioffe07}. This means the qubit-flip and
phase-shift errors happen with different probabilities. Therefore,
it is needed to construct quantum codes that deal with the realistic
quantum noise. We derive families  of asymmetric quantum error
control codes that differentiate between these two kinds of errors,
$\Pr{Z}>\Pr{X}$.

\begin{definition}[AQEC]
A $q$-ary asymmetric quantum code $Q$, denoted by
$[[n,k,d_z/d_x]]_q$, is a $q^k$ dimensional subspace of the Hilbert
space $\mathbb{C}^{q^n}$ and can control all bit-flip errors up to
$\lfloor \frac{d_x-1}{2}\rfloor$ and all phase-flip errors up to
$\lfloor \frac{d_z-1}{2}\rfloor$. The code $Q$ detects $(d_1-1)$
qubit-flip errors as well as detects $(d_1-1)$ phase-shift errors.
\end{definition}

We use different notation from the one given in~\cite{evans07}. The
reason is that we would like to compare $d_z$ and $d_x$ as a factor
$\rho =d_z/d_x$ not as a ratio. Therefore, if $d_z>d_x$, then the
AQEC has a factor great than one. Hence, the  phase-shift errors
affect the quantum system more than qubit-flip errors do.  In our
work, we would like to increase both the  factor $\rho$ and
dimension $k$ of the quantum code.

\bigbreak

 \noindent \textbf{Connection to Classical nonbinary Codes.} Let $C_1$ and $C_2$ be two linear codes over the finite field $\F_q$, and
let $[n,k_1,d_1]_q$ and $[n,k_2,d_2]_q$ be their parameters. For
$i\in \{1,2\}$, if $H_i$  is the parity check matrix of the code
$C_i$, then $\dim{C_i^{\perp}}=n-k_i$ and rank of $H_i^\perp$ is
$k_i$. If $C_{i}^\perp \subseteq C_{1+(i\mod 2)}$, then $C_{1+(i
\mod 2)}^\perp \subseteq C_i$. So, the rows of $H_i$ which form a
basis for $C_i^\perp$ can be extended to form a basis for
$C_{1+(i\mod 2)}$ by adding some vectors. Also, if $g_i(x)$ is the
generator polynomial of a cyclic code $C_i$ then
$k_i=n-deg(g_i(x))$, see~\cite{macwilliams77,huffman03}.

The error groups $\G_x$ and $\G_z$ can be mapped, respectively,  to
two classical codes $C_1$ and $C_2$ in  a similar manner as in QEC.
This connection is well-know, see for
example~\cite{calderbank98,rains99,sarvepalli07a}. Let $C_i$ be a
classical code such that $C_{1+(i\mod 2)}^\perp \subseteq C_i$ for
$i \in \{1,2\}$, then we have a symmetric quantum control code
(AQEC) with parameters $[[n,k_1+ k_2-n,d_z /d_x]]_q$. This can be
illustrated in the following result.

\begin{figure}[t]
  \begin{center}
  \includegraphics[scale=0.65]{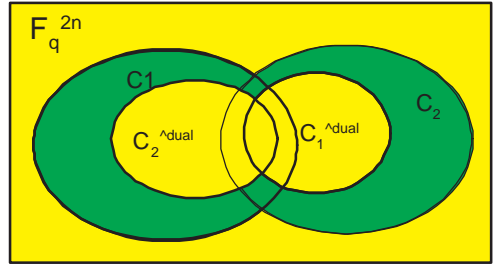}
  \caption{Constructions of asymmetric quantum codes based on two classical codes $C_1$ and $C_2$
  with parameters $[n,k_1]$ and $[n,d_2]$ such that $C_i \subseteq C_{1+(i \mod 2)}$ for $i=\{1,2\}$. AQEC has parameters $[[n,k_1+k_2-n,d_z/d_x]]_q$ where
  $d_x=\wt(C_1 \backslash C_2^\perp)$ and $d_z=\wt(C_2 \backslash C_1^\perp)$}\label{fig:subsys1}
  \end{center}
\end{figure}

\begin{lemma}[CSS AQEC]\label{lem:AQEC}
Let $C_i$ be a classical code with parameters $[n,k_i,d_i]_q$ such
that $C_i^\perp \subseteq C_{1+(i\mod 2)}$ for $i \in \{1,2\}$ , and
$d_x= \min\big\{\wt(C_{1} \backslash C_2^\perp), \wt(C_{2}
\backslash C_{1 }^\perp)\big\}$, and $d_z= \max\big\{\wt(C_{1}
\backslash C_2^\perp), \wt(C_{2} \backslash C_{1 }^\perp)\big\}$.
Then there is asymmetric quantum code with parameters
$[[n,k_1+k_2-n,d_z/d_x]]_q$. The quantum code is pure to its minimum
distance meaning that if $\wt(C_1)=\wt(C_1\backslash C_2^\perp)$
then the code is pure to $d_x$, also if $\wt(C_2)=\wt(C_2\backslash
C_1^\perp)$ then the code is pure to $d_z$.
\end{lemma}

Therefore, it is straightforward to derive asymmetric quantum
control codes from two classical codes as shown in
Lemma~\ref{lem:AQEC}. Of course, one wishes to increase the values
of $d_z$ vers. $d_x$ for the same code length and dimension.
\begin{remark}
The notations of purity and impurity of AQEC remain the same as
shown for QEC, the interested reader might consider any primary
papers on QEC.
\end{remark}

\section{Asymmetric Quantum BCH and RS Codes}\label{sec:AQEC-BCH}
In this section we derive classes of AQEC based on classical BCH and
RS codes. We will restrict ourself to the Euclidean construction for
codes defined over $\F_q$. However, the generalization to the
Hermitian construction for codes defined over $\F_{q^2}$ is straight
forward. We keep the definitions of BCH codes to a minimal since
they have been well-known, see  example~\cite{aly07a} or any
textbook on classical coding
theory~\cite{macwilliams77,huffman03,hocquenghem59}. Let $q$ be a
power of a prime and $n$ a positive integer such that $\gcd(q,n)=1$.
Recall that the cyclotomic coset $S_x$ modulo $n$ is defined as
\begin{eqnarray}S_x=\{xq^i\bmod n \mid i\in \Z, i\ge 0\}.
\end{eqnarray}

Let $m$ be the multiplicative order of $q$ modulo $n$. Let $\alpha$
be a primitive element in $\F_{q^m}$. A nonprimitive narrow-sense
BCH code $C$ of designed distance $\delta$ and length $n$ over
$\F_{q}$ is a cyclic code with a generator monic polynomial $g(x)$
that has $\alpha, \alpha^2, \ldots, \alpha^{\delta-1}$ as zeros,
\begin{eqnarray}
g(x)=\prod_{i=1}^{\delta -1} (x-\alpha^i).
\end{eqnarray}
Thus,  $c$ is a codeword in $\mathcal{C}$ if and only if
$c(\alpha)=c(\alpha^2)=\ldots=c(\alpha^{\delta-1})=0$. The parity
check matrix of this code can be defined as
\begin{eqnarray}\label{bch:parity}
 H_{bch} =\left[ \begin{array}{ccccc}
1 &\alpha &\alpha^2 &\cdots &\alpha^{n-1}\\
1 &\alpha^2 &\alpha^4 &\cdots &\alpha^{2(n-1)}\\
\vdots& \vdots &\vdots &\ddots &\vdots\\
1 &\alpha^{\delta-1} &\alpha^{2(\delta-1)} &\cdots
&\alpha^{(\delta-1)(n-1)}
\end{array}\right].
\end{eqnarray}

In general the dimensions and minimum distances of BCH codes are not
known. However,  lower bounds on these two parameters for such codes
are  given by $d \geq \delta$ and $k \geq n-m(\delta-1)$.
Fortunately, in~\cite{aly07a,aly06a} exact formulas for the
dimensions and minimum distances are given under certain conditions.
The following result shows the dimension of BCH codes.

\medskip

\begin{theorem}[Dimension BCH Codes]\label{th:bchnpdimension}
Let $q$ be a prime power and $\gcd(n,q)=1$, with $ord_n(q)=m$. Then
a narrow-sense BCH code of length $q^{\lfloor m/2\rfloor} <n \leq
q^m-1$ over $\F_q$ with designed distance $\delta$ in the range $2
\leq \delta \le \delta_{\max}= \min\{ \lfloor nq^{\lceil m/2
\rceil}/(q^m-1)\rfloor,n\}$, has dimension of
\begin{equation}\label{eq:npdimension}
k=n-m\lceil (\delta-1)(1-1/q)\rceil.
\end{equation}
\end{theorem}
\begin{proof}
See~\cite[Theorem 10]{aly07a}.
\end{proof}
Steane first derived binary quantum BCH codes
in~\cite{steane96,steane99}. In addition Grassl \emph{el. at.} gave
a family of quantum BCH codes along with tables of best
codes~\cite{grassl99b}.

In~\cite{aly06a,aly07a}, while it was a challenging task to derive
self-orthogonal or dual-containing conditions for BCH codes, we can
relax and omit these conditions by looking for  BCH codes that are
nested. The following result shows a family of QEC derived from
nonprimitive narrow-sense BCH codes.

We can also switch between the code and its dual to construct a
quantum code.  When the BCH codes contain their duals, then we can
derive the following codes.

\begin{theorem}\label{sh:euclid}
Let $m=\ord_n(q)$  and $q^{\lfloor m/2\rfloor} <n \leq q^m-1$ where
$q$ is a power of a prime and $2\le \delta\le \delta_{\max},$ with
$$\delta_{\max}^*=\frac{n}{q^m-1}(q^{\lceil m/2\rceil}-1-(q-2)[m  \textup{ odd}]),$$
then there exists a quantum code with parameters
$$[[n,n-2m\lceil(\delta-1)(1-1/q)\rceil,\ge \delta]]_q$$ pure to $\delta_{\max}+1$
\end{theorem}
\begin{proof}
See~\cite[Theorem 19]{aly07a}.
\end{proof}

\medskip

\subsection{AQEC-BCH}
Fortunately, the mathematical structure of BCH codes always us
easily to show the nested required structure as needed in
Lemma~\ref{lem:AQEC}. We know that $g(x)$ is a generator polynomial
of a narrow sense BCH code that has roots
$\alpha^2,\alpha^3,\ldots,\alpha^{\delta-1}$ over $\F_{q}$. We know
that the generator polynomial has degree $m \lfloor
(\delta-1)(1-1/\delta)\rfloor$ if $\delta \leq \delta_{max}$.
Therefore the dimension is given by $k=n-deg(g(x))$. Hence, the
nested structure of BCH codes is obvious and can be described as
follows. Let
\begin{eqnarray}\delta_{i+1} > \delta_i > \delta_{i-1} \geq \ldots \geq 2,\end{eqnarray}
and let $C_i$ be a BCH code that has generator polynomial $g_i(x)$,
in which it has roots $\{2,3,\ldots,\delta-1\}$. So, $C_i$ has
parameters $[n,n-deg(g_i(x)),d_i\geq \delta_i]_q$, then
\begin{eqnarray}
C_{i+1} \subseteq C_{i} \subseteq C_{i-1} \subseteq \ldots
\end{eqnarray}

We need to ensure that $\delta_i$ and $\delta_{i+1}$ away of each
other, so the elements (roots) $\{2,\ldots,\delta_i-1\}$ and
$\{2,\ldots,\delta_{i+1}-1\}$ are different. This means that the
cyclotomic cosets generated by $\delta_i$ and $\delta_{i+1}$ are not
the same, $S_1\cup \ldots \cup S_{\delta_i-1} \neq S_1\cup \ldots
\cup S_{\delta_{i+1}-1} $. Let $\delta_i^\perp$ be the designed
distance of the code $C_i^\perp$. Then the following result gives a
family of AQEC BCH codes over $\F_q$.

\begin{table}[t]
\caption{Families of asymmetric quantum BCH codes~\cite{magma}}
\label{table:bchtable}
\begin{center}
\begin{tabular}{|c|c|c|c|c|}
\hline   q & $C_1$ BCH Code & $C_2$ BCH Code &AQEC \\
 \hline
 &&&\\
 2&$[15,11,3]$&$[15,7,5]$&$[[15,3,5/3]]_2$\\
 2&$[15,8,4]$&$[15,7,5]$&$[[15,0,5/4]]_2$\\
 2&$[31, 21, 5]$ & $[31, 16, 7]$& $[[31,6, 7/5]]_2$\\
 2&$[31,26,3]$&$[31,16,7]$&$[[31,11,7/3]]$\\
 2&$[31,26,3]$&$[31,16,7]$&$[[31,10,8/3]]$\\
 2&$[31,26,3]$&$[31,11,11]$&$[[31,6,11/3]]$\\
  2&$[31,26,3]$&$[31,6,15]$&$[[31,1,15/3]]$\\
2&$[127,113,5]$&$[127,78,15]$&$[[127,64,15/5]]$\\
2&$[127,106,7]$&$[127,77,27]$&$[[127,56,25/7]]$\\

  \hline
\end{tabular}
\end{center}
\end{table}
%


\begin{theorem}[AQEC-BCH]\label{thm:AQEC-bch}
Let $q$ be a prime power and $\gcd(n,q)=1$, with $ord_n(q)=m$. Let
$C_1$ and $C_2$ be two  narrow-sense BCH codes of length $q^{\lfloor
m/2\rfloor} <n \leq q^m-1$ over $\F_q$ with designed distances
$\delta_1$ and $\delta_2$ in the range $2 \leq \delta_1, \delta_2
 \le \delta_{\max}= \min\{ \lfloor nq^{\lceil m/2
\rceil}/(q^m-1)\rfloor,n\}$ and $\delta_1 <\delta_2^\perp \leq
\delta_2 <\delta_1^\perp$.

Assume $S_1\cup \ldots \cup S_{\delta_1-1} \neq S_1\cup \ldots \cup
S_{\delta_{2}-1}$, then there exists an asymmetric quantum error
control code with parameters $[[n,n-m\lceil
(\delta_1-1)(1-1/q)\rceil-m\lceil (\delta_2-1)(1-1/q)\rceil,\geq
d_z/d_x]]_q$, where $d_z=\wt(C_2 \backslash C_1^\perp) \geq \delta_2
> d_x=\wt(C_1 \backslash C_2^\perp) \geq \delta_1$.
\end{theorem}

\begin{proof}
From the nested structure of BCH codes, we know that if $\delta_1 <
\delta_2^\perp$, then $C_2^\perp \subseteq C_1$, similarly if
$\delta_2 < \delta_1^\perp$, then $C_1^\perp \subseteq C_2$. By
Lemma~\ref{th:bchnpdimension}, using the fact that $\delta \leq
\delta_{\max}$, the dimension of the code $C_i$ is given by
$k_i=n-m\lceil (\delta_i-1)(1-1/q)\rceil$ for $i=\{1,2\}$. Since
$S_1\cup \ldots \cup S_{\delta_1-1} \neq S_1\cup \ldots \cup
S_{\delta_{2}-1}$, this means that $deg(g_1(x)) < deg(g_2(x))$,
hence $k_2 <k_1$. Furthermore $k_1^\perp < k_2^\perp$.

By Lemma~\ref{lem:AQEC} and  we  assume $ d_x=wt(C_1 \backslash
C_2^\perp) \geq \delta_1$ and $ d_z=wt(C_2 \backslash C_1^\perp)
\geq \delta_2$ such that $d_z>d_x$ otherwise we exchange the rules
of $d_z$ and $d_x$; or the code $C_i$ with $C_{1+(i\mod 2)}$.
Therefore, there exists AQEC with parameters $[[n,k_1+k_2-n,\geq
d_z/d_z]]_q$.
\end{proof}

The problem with BCH codes is that we have lower bounds on their
minimum distance given their arbitrary designed distance. We argue
that their minimum distance meets with their designed distance for
small values that are particularly interesting to us. One can also
use the condition shown in~\cite[Corollary 11.]{aly07a} to ensure
that the minimum distance meets the designed distance.

The condition regarding the designed distances $\delta_1$ and
$\delta_2$ allows us to give  formulas for the dimensions of BCH
codes $C_1$ and $C_2$, however, we can derive AQEC-BCH without this
condition as shown in the following result. This is  explained by an
example in the next section.

\begin{lemma}\label{thm:AQEC-bch2}
Let $q$ be a prime power, $\gcd(m,q)=1$, and $q^{\lfloor m/2\rfloor}
<n \leq q^m-1$  for some integers $m=\ord_n(q)$. Let $C_1$ and $C_2$
be two BCH codes with parameters $[n,k_1,d_x \geq \delta_1]_q$ and
$[n,k_2,d_z \geq \delta_2]_q$, respectively, such that $\delta_{1}
<\delta_2^\perp \leq \delta_2 <\delta_1^\perp$, and $k_1+k_2 >n$.
Assume $S_1\cup \ldots \cup S_{\delta_1-1} \neq S_1\cup \ldots \cup
S_{\delta_{2}-1}$, then there exists an asymmetric quantum error
control code with parameters $[[n,k_1+k_2-n, \geq d_z/d_x]]_q$,
where $d_z=\wt(C_1 \backslash C_2^\perp)=\delta_2 >  d_x=\wt(C_2
\backslash C_1^\perp)=\delta_1$.
\end{lemma}
In fact the previous theorem can be used to derive any asymmetric
cyclic quantum control codes. Also, one can construct AQEC based on
codes that are defined over $\F_{q^2}$.

\subsection{RS Codes} We can also derive a family of asymmetric
quantum control codes based on Redd-Solomon codes. Recall that a RS
code with length $n=q-1$ and designed distance $\delta$ over a
finite field $\F_q$ is a code with parameters $[[n,n-d+1
,d=\delta]]_q$ and generator polynomial
\begin{eqnarray}
g(x)=\prod_{i=1}^{d-1}(x-\alpha^i).
\end{eqnarray}
It is much easier to derive conditions for AQEC derived from RS as
shown in the following theorem.

\begin{theorem}
Let $q$ be a prime power and $n=q-1$. Let $C_1$ and $C_2$ be two RS
codes with parameters $[n,n-d_1+1,d_1]]_q$ and $[n,n-d_2+1,d_2]_q$
for $d_1<d_2<d_1^\perp=n-d_1$. Then there exists AQEC code with
parameters $[[n, n-d_1-d_1+2,d_z/d_x]]_q$, where $d_x=d_1 <d_z=d_2$.
\end{theorem}
\begin{proof}
since $d_1<d_2<d_1^\perp$, then $n-d_1^\perp+1<n-d_2+1<n-d_1+1$ and
$k_1^\perp<k_2 < k_1$. Hence $C_2^\perp \subset C_1$ and $C_1^\perp
\subset C_2$. Let $d_z = \wt(C_2\backslash C_1^\perp)= d_2$ and $d_x
= \wt(C_1\backslash C_2^\perp)= d_1$. Therefore there must exist
AQEC with parameters $[[n,n-d_1-d_1+2,d_z/d_x]]_q$.
\end{proof}
It is obvious from this theorem that the constructed code is a pure
code to its minimum distances.  One can also derive asymmetric
quantum RS codes based on RS codes over $\F_{q^2}$. Also,
generalized RS codes can be used to derive similar results. In fact,
one can derive AQEC from any two classical cyclic codes obeying the
pair-nested structure over $\F_q$.

\section{AQEC and Connection with Subsystem
Codes}\label{sec:AQEC-subsystem}

In this section we establish the connection between AQEC and
subsystem codes. Furthermore we derive a larger class of quantum
codes called asymmetric subsystem codes (ASSs). We derive families
of subsystem BCH codes and cyclic subsystem codes over $\F_q$.
In~\cite{aly08f,aly08a} we construct several families of subsystem
cyclic, BCH, RS and MDS codes over $\F_{q^2}$ with much more details

We expand our understanding of the theory of quantum error control
codes by correcting the quantum errors $X$ and $Z$ separately using
two different classical codes, in addition to correcting only errors
in a small subspace. Subsystem codes are a generalization of the
theory of quantum error control codes, in which errors can be
corrected as well as avoided (isolated).

\begin{figure}[t]
  \begin{center}
  \includegraphics[scale=0.65]{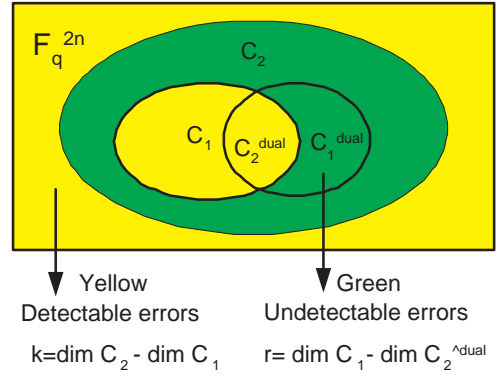}
  \caption{A quantum code Q is decomposed into two subsystem A (info) and B (gauge)}\label{fig:subsys1}
  \end{center}
\end{figure}

Let $Q$ be a quantum code  such that $\mathcal{H}=Q\oplus Q^\perp$,
where $Q^\perp$ is the orthogonal complement of $Q$. We can define
the subsystem code $Q=A\otimes B$, see Fig.\ref{fig:subsys1}, as
follows
\begin{definition}[Subsystem Codes]
An $[[n,k,r,d]]_q$ subsystem code is a decomposition of the subspace
$Q$ into a tensor product of two vector spaces A and B such that
$Q=A\otimes B$, where $\dim A=q^k$ and $\dim B= q^r$. The code $Q$
is able to detect all errors  of weight less than $d$ on subsystem
$A$.
\end{definition}
Subsystem codes can be constructed  from the classical codes  over
$\F_q$ and $\F_{q^2}$. Such codes do not need the classical codes to
be self-orthogonal (or dual-containing) as shown in the Euclidean
construction. We have given  general constructions of subsystem
codes in~\cite{aly06c} known as the subsystem CSS and Hermitian
Constructions. We provide a proof for the following special case of
the CSS construction.

\begin{lemma}[SSC Euclidean Construction]\label{lem:css-Euclidean-subsys}
If $C_1$ is a $k'$-dimensional $\F_q$-linear code of length $n$ that
has a $k''$-dimensional subcode $C_2=C_1\cap C_1^\perp$ and
$k'+k''<n$, then there exist
$$[[n,n-(k'+k''),k'-k'',\wt(C_2^\perp\setminus C_1)]]_q$$
$$[[n,k'-k'',n-(k'+k''),\wt(C_2^\perp\setminus C_1)]]_q$$
subsystem codes.
\end{lemma}
\begin{IEEEproof}
Let us define the code $X=C_1\times C_1 \subseteq \F_q^{2n}$,
therefore $X^\sdual=(C_1\times C_1)^\sdual=C_1^\sdual\times
C_1^\sdual$. Hence $Y=X \cap X^\sdual=(C_1\times C_1)\cap
(C_1^\sdual\times C_1^\sdual)= C_2 \times C_2$. Thus, $\dim_{\F_q}
Y=2k''$. Hence $|X||Y|=q^{2(k'+k'')}$ and $|X|/|Y|=q^{2(k'-k'')}$.
By Theorem~\cite[Theorem 1]{aly06c}, there exists a subsystem code
$Q=A\otimes B$ with parameters $[[n,\log_q\dim A,\log_q\dim B,
d]]_q$ such that
\begin{compactenum}[i)]
\item $\dim A=q^n/(|X||Y|)^{1/2}=q^{n-k'-k''}$.
\item $\dim B=(|X|/|Y|)^{1/2}=q^{k'-k''}$.
\item $d= \swt(Y^\sdual \backslash X)=\wt(C_2^\perp\setminus C_1)$.
\end{compactenum}
Exchanging the rules of the codes $C_1$ and $C_1^\perp$ gives us the
other subsystem code with the given parameters.
\end{IEEEproof}

Subsystem codes (SCC) require the code $C_2$ to be self-orthogonal,
$C_2 \subseteq C_2^\perp$. AQEC and SSC are both can be constructed
from the pair-nested classical codes, as we call them. From this
result, we can see that any two classical codes $C_1$ and $C_2$ such
that $C_2=C_1 \cap C_1^\perp \subseteq C_2^\perp$, in which they can
be used to construct a subsystem code (SCC), can be also used to
construct asymmetric quantum code (AQEC). Asymmetric subsystem codes
(ASSCs) are much larger class than  the class of symmetric subsystem
codes, in which the quantum errors occur with different
probabilities in the former one and have equal probabilities in the
later one. In short, AQEC does does not require the intersection
code to be self-orthogonal.

The construction in Lemma~\ref{lem:css-Euclidean-subsys} can be
generalized to ASSC CSS construction in a similar way. This means
that we can look at an AQEC with parameters $[[n,k,d_z/d_x]]_q$. as
subsystem code with parameters $[[n,k,0,d_z/d_x]]_q$. Therefore all
results shown in~\cite{aly08a,aly08f,aly06c} are a direct
consequence by just fixing the minimum distance condition.

We have shown in~\cite{aly08a,aly08f} that All stabilizer codes
(pure and impure) can be reduced to subsystem codes as shown in the
following result.
\goodbreak
\begin{theorem}[Trading Dimensions of SSC and Co-SCC]\label{th:FqshrinkK}
Let $q$ be a power of a prime~$p$. If there exists an $\F_q$-linear
$[[n,k,r,d]]_q$ subsystem code (stabilizer code if $r=0$) with $k>1$
that is pure to $d'$, then there exists an $\F_q$-linear
$[[n,k-1,r+1,\geq d]]_q$ subsystem code that is pure to
$\min\{d,d'\}$.  If a pure ($\F_q$-linear) $[[n,k,r,d]]_q$ subsystem
code exists, then a pure ($\F_q$-linear) $[[n,k+r,d]]_q$ stabilizer
code exists.
\end{theorem}

We have shown in~\cite{aly07a,aly06a} that narrow sense BCH codes,
primitive and non-primitive, with length $n$ and designed distance
$\delta$ are Euclidean dual-containing codes if and only if
\begin{eqnarray}
2\! \le\!  \delta\le \delta_{\max}\!= \!\frac{n}{q^{m}-1} (q^{\lceil
m/2\rceil}\!-\!1\!-\!(q-2)[m \textup{ odd}]).
\end{eqnarray}  We
use this result and~\cite[Theorem 2]{aly08a} to derive nonprimitive
subsystem BCH codes from classical BCH codes over $\F_q$ and
$\F_{q^2}$~\cite{aly06c,aly06a}. The subsystem codes derived
in~\cite{aly08f} are only for the primitive case.

\begin{lemma}\label{lem:BCHExistFq}
If $q$ is  power of a prime, $m$ is a positive integer, and
$q^{\lfloor m/2\rfloor} <n \leq q^m-1$. Let $2\le \delta\le
\delta_{\max}=\frac{n}{q^{m}-1} (q^{\lceil m/2\rceil}-1-(q-2)[m
\textup{ odd}])$, then there exists a subsystem BCH code with
parameters $[[n,n-2m\lceil(\delta-1)(1-1/q) \rceil -r, r,\geq \delta
]]_q$ where $0 \leq r< n-2m\lceil(\delta-1)(1-1/q) \rceil$.
\end{lemma}
\begin{proof}
We know that if $2\le \delta\le \delta_{\max}=\frac{n}{q^{m}-1}
(q^{\lceil m/2\rceil}-1-(q-2)[m \textup{ odd}])$, the the classical
BCH codes contain their Euclidean dual code by~\cite[Theorem
3.]{aly07a}. But existence of this code gives  a stabilizer code
with parameters $[[n,n-2m\lceil(\delta-1)(1-1/q) \rceil, \geq \delta
]]_q$ using \cite[Theorem 19.]{aly07a}.

We know that every stabilizer code can be reduced to a subsystem
code by Theorem~\ref{th:FqshrinkK}. Let r be an integer in the range
$0 \leq r< n-2m\lceil(\delta-1)(1-1/q) \rceil$. From~\cite[Theorem
2]{aly08a} or Theorem~\ref{th:FqshrinkK}, then there must exist a
subsystem BCH code with parameters $ [[n,n-2m\lceil(\delta-1)(1-1/q)
\rceil -r, r,\geq \delta ]]_q$.
\end{proof}

\medskip

We can also construct subsystem BCH codes from stabilizer codes
using the Hermitian constructions where the classical BCH codes are
defined over $\F_{q^2}$.

\begin{lemma}\label{lem:BCHExistFq2}
If $q$ is a power of a prime, $m=\ord_n(q^2)$ is a positive integer,
and $\delta$ is an integer in the range $2\le \delta \le
\delta_{\max}=\lfloor n(q^m-1)/(q^{2m}-1)\rfloor$, then there exists
a subsystem code $Q$ with parameters
$$ [[n, n-2m\lceil(\delta-1)(1-1/q^2)\rceil -r,r, d_Q\ge
\delta]]_q$$ that is pure up to $\delta$, where $0 \leq r
<n-2m\lceil(\delta-1)(1-1/q^2)\rceil$.
\end{lemma}
\begin{proof}
We knot that if $2\le \delta \le \delta_{\max}=\lfloor
n(q^m-1)/(q^{2m}-1)\rfloor$, then exists a classical BCH code with
parameters $[n,n-m\lceil(\delta-1)(1-1/q^2)\rceil,\ge \delta]_q$
which contains its Hermitian dual code using~\cite[Theorem
14.]{aly07a}. But existence of the classical code that contains its
Hermtian code gives us quantum codes by~\cite[Theorem 21.]{aly07a}.
From~\cite[Theorem 2]{aly08a}, then there must exist a subsystem
code with the given parameters $ [[n,
n-2m\lceil(\delta-1)(1-1/q^2)\rceil -r,r, d_Q\ge \delta]]_q$ that is
pure up to $\delta$, for all range of $r$  in $0 \leq r
<n-2m\lceil(\delta-1)(1-1/q^2)\rceil$..
\end{proof}

If fact there is  a tradeoff between the construction of subsystem
codes and asymmetric quantum codes. The condition $C_2=C_1\cap
C_1^\perp$ used for the construction of SSC, is not needed in the
construction of AQEC.

\begin{table}[t]
\caption{subsystem BCH codes using the Euclidean Construction}
\label{table:bchtable}
\begin{center}
\begin{tabular}{|l|l|c|}
\hline
\text{Subsystem Code} &  \text{Parent} & \text {Designed}  \\
 &  \text{BCH Code} & \text{ distance }  \\
 \hline
 &&\\
 $[[15 ,4 ,3 ,3 ]]_2$   &$[15 ,7 ,5 ]_2$  & 4\\
 $ [[15 ,6 ,1 ,3 ]]_2  $ &$[15 ,5 ,7 ]_2 $ & 6\\
  $ [[31 ,10,1 ,5 ]]_2 $  &$[31 ,11,11]_2 $ & 8\\
  $  [[31 ,20,1 ,3 ]]_2  $ &$[31 ,6 ,15]_2 $ & 12\\
   $  [[63 ,6 ,21,7 ]]_2 $  &$[63 ,39,9 ]_2 $ & 8\\
$ [[63 ,6 ,15,7 ]]_2 $  &$[63 ,36,11]_2$  & 10\\
 $ [[63 ,6 ,3 ,7 ]]_2 $  &$[63 ,30,13]_2$  & 12\\
$ [[63 ,18,3 ,7 ]]_2$   &$[63 ,24,15]_2$  & 14\\
$  [[63 ,30,3 ,5 ]]_2$   &$[63 ,18,21]_2 $ & 16\\
  $ [[63 ,32,1 ,5 ]]_2 $  &$[63 ,16,23]_2 $ & 22\\
  $  [[63 ,44,1 ,3 ]]_2  $ &$[63 ,10,27]_2 $ & 24\\
  $  [[63 ,50,1 ,3 ]]_2  $ &$[63 ,7 ,31]_2  $& 28\\
  \hline
  &&\\
  $[[15 ,2 ,5 ,3 ]]_4$   &$[15 ,9 ,5 ]_4$  & 4\\
   $[[15 ,2 ,3 ,3 ]]_4  $ &$[15 ,8 ,6 ]_4$  & 6\\
    $[[15 ,4 ,1 ,3 ]]_4  $ &$[15 ,6 ,7 ]_4$  & 7\\
     $[[15 ,8 ,1 ,3 ]]_4  $ &$[15 ,4 ,10]_4$  & 8\\

$[[31 ,10,1 ,5 ]]_4  $ &$[31 ,11,11]_4$  & 8\\
 $[[31 ,20,1 ,3 ]]_4  $ &$[31 ,6 ,15]_4$  & 12\\
  $[[63 ,12,9 ,7 ]]_4  $ &$[63 ,30,15]_4$  & 15\\
   $[[63 ,18,9 ,7 ]]_4  $ &$[63 ,27,21]_4$  & 16\\
    $[[63 ,18,7 ,7 ]]_4  $ &$[63 ,26,22]_4$  & 22\\

 \hline
\end{tabular}
\\$*$ punctured code\\
$+$ Extended code
\end{center}
\end{table}


Instead of constructing subsystem codes from stabilizer BCH codes as
shown in Lemmas~\ref{lem:BCHExistFq}, \ref{lem:BCHExistFq2},  we can
also construct subsystem codes from classical BCH codes over $\F_q$
and $\F_{q^2}$ under some restrictions on the designed distance
$\delta$. Let $S_i$ be a cyclotomic coset defined as $\{ iq^j \mod n
\mid j \in \Z \}$. We will derive only SSC from nonprimitive BCH
codes over $\F_q$; for codes over $\F_{q^2}$ and further details
see~\cite{aly08f}. Also, the generator polynomial can be used
instead of the defining set (cylotomic cosets) to derive BCH codes.


\begin{lemma}\label{lem:subsysBCHq}
If $q$ is a power of a prime, $m=\ord_n(q)$  is a positive integer
and  $2\le \delta\le\delta_{\max}=\frac{n}{q^{m}-1} (q^{\lceil
m/2\rceil}-1-(q-2)[m\textup{ odd}])$.  Let $C_2$ be a BCH code with
length $q^{\lfloor m/2\rfloor} <n \leq q^m-1$ and defining set
$T_{C_2}=\{S_0,S_1,\ldots, S_{n-\delta} \}$, such that
$\gcd(n,q)=1$. Let $T \subseteq \{ 0\}\cup\{S_\delta, \ldots,
S_{n-\delta} \}$ be a nonempty set. Assume $C_1 \subseteq \F_q^n$ be
a BCH code with the defining set $T_{C_1}=\{S_0,S_1,\ldots,
S_{n-\delta} \}\setminus (T\cup T^{-1})$ where $T^{-1} =\{ -t \bmod
n\mid t\in T\}$. Then there exists a subsystem BCH code with the
parameters $[[n,n-2k-r,r,\geq \delta]]_q$, where
$k=m\lceil(\delta-1)(1-1/q) \rceil$ and $0\leq r=|T\cup
T^{-1}|<n-2k$.
\end{lemma}
\begin{proof}
The proof can be divided into the following parts:

\begin{enumerate}[i)]
\item
We know that $T_{C_2}=\{S_0,S_1,\ldots, S_{n-\delta} \}$ and $T
\subseteq \{ 0\}\cup\{S_\delta, \ldots, S_{n-\delta} \}$ be a
nonempty set. Hence $T_{C_2}^\perp=\{S_1,\ldots, S_{\delta-1}
\}$. Furthermore, if $2\le
\delta\le\delta_{\max}=\frac{n}{q^{m}-1} (q^{\lceil
m/2\rceil}-1-(q-2)[m\textup{ odd}])$, then $C_2 \subseteq
C_2^\perp$. Furthermore, let $k=m\lceil(\delta-1)(1-1/q)
\rceil$, then $\dim C_2^\perp=n-k$ and $\dim C_2= k$.

\item  We know that $C_1 \in \F_q^n$ is a BCH code with defining set
$T_{C_1}=T_{C_2} \setminus (T\cup T^{-1})=\{S_0,S_1,\ldots,
S_{n-\delta} \}\setminus (T\cup T^{-1})$ where $T^{-1} =\{ -t
\bmod n\mid t\in T\}$. Then the dual code $C_1^\perp$ has
defining set $T_{C_1}^\perp=\{S_1,\ldots, S_{\delta-1} \}\cup
T\cup T^{-1}=T_{C_2^\perp}\cup T\cup T^{-1}$. We can compute the
union set $T_{C_2}$  as $T_{C_1} \cup
T_{C_1}^\perp=\{S_0,S_1,\ldots, S_{n-\delta} \}=T_{C_2}$.
Therefore, $C_1 \cap C_1^\perp=C_2$. Furthermore, if $0\leq
r=|T\cup T^{-1}| <n-2k$, then $\dim C_1= k+r$.

\item
From step (i) and (ii), and for $0\leq r <n-2k$, and by
Lemma~\ref{lem:css-Euclidean-subsys},  there exits a subsystem
code with parameters $[[n,\dim C_2^\perp-\dim C_1,\dim C_1-\dim
C_2, d]]_q=[[n,n-2k-r,r,d]]_q$, $d=\min wt(C_2^\perp - C_1)\geq
\delta$.
\end{enumerate}
\end{proof}

One can also construct asymmetric subsystem BCH codes in a natural
way meaning the distances $d_x$ and $d_z$ can be defined using the
AQEC definition. In other words one can obtain ASSCs with parameters
$[[n,n-2k-r,r,d_z/d_x]]_q$ and $[[n,r,n-2k-r,d_z/d_x]]_q$.  The
extension to ASSCs based on RS codes is straight forward and similar
to our constructions in~\cite{aly08a,aly08f}.

\subsection{Cyclic Subsystem Codes}

Now, we shall give a general construction for subsystem cyclic
codes. This would apply for all cyclic codes including BCH, RS, RM
and duadic codes.  We show that if a classical cyclic code is
self-orthogonal, then one can easily construct cyclic subsystem
codes. We say that a code $C_2$ is self-orthogonal if and only if
$C_2\subseteq C_2^\perp$. We will derive subsystem cyclic codes over
$\F_q$, and the case of $\F_{q^2}$ is illustrated in~\cite{aly08f}.

\begin{theorem}\label{lem:cyclic-subsysI}
Let $C_2$ be a $k$-dimensional self-orthogonal cyclic code of length
$n$ over $\F_q$. Let $T_{C_2}$ and $T_{C_2^\perp}$ respectively
denote the defining sets of $C_2$ and $C_2^\perp$. If $T$ is a
subset of $T_{C_2} \setminus T_{C_2^\perp}$ that is the union of
cyclotomic cosets, then one can define a cyclic code $C_1$ of length
$n$ over $\F_q$ by the defining set $T_{C_1}= T_{C_2} \setminus (T
\cup T^{-1})$.  If $r=|T\cup T^{-1}|$ is in the range $0\le r<
n-2k$, and $d= \min \wt(C_2^\perp \setminus C)$, then there exists a
subsystem code with parameters $[[n,n-2k-r,r,d]]_q$.
\end{theorem}
\begin{proof}
see~\cite{aly08f} and more details are shown in
in~\cite{aly08thesis}.
\end{proof}

Now it is straight forward to derive asymmetric cyclic subsystem
codes with parameters $[[n,n-2k-r,r,d_z/d_x]]_q$ for all $0 \leq r
<n-2k$ using Theorem~\ref{lem:cyclic-subsysI} where $d_x= \min
\{\wt(C_2^\perp \backslash C_1), \wt(C_2^\perp \backslash C_1^\perp)
\}$ and $d_z= \max \{\wt(C_1^\perp \backslash C_2), \wt(C_1^\perp
\backslash C_2) \}$.
\section{Illustrative Example}\label{sec:example}
We have demonstrated a family  of asymmetric quantum codes with
arbitrary length, dimension, and minimum distance parameters. We
will present a simple example to explain our construction.

Consider a BCH code $C_1$ with parameters $[15,11,3]_2$  that has
designed distance $3$ and generator matrix given by

\begin{eqnarray} \left[\begin{array}{p{0.1cm}p{0.1cm}p{0.1cm}cc ccccc cccccc}
1& 0& 0& 0& 0& 0& 0& 0& 0& 0& 0& 1& 1& 0& 0\\
0& 1& 0& 0& 0& 0& 0& 0& 0& 0& 0& 0& 1& 1& 0\\
0& 0& 1& 0& 0& 0& 0 &0 &0& 0& 0& 0& 0& 1& 1\\0& 0& 0& 1& 0 &0 &0& 0&
0& 0& 0& 1& 1& 0& 1\\0& 0& 0& 0& 1& 0& 0& 0& 0& 0& 0& 1& 0& 1& 0\\0
&0& 0& 0& 0& 1& 0& 0& 0& 0& 0& 0& 1& 0& 1\\0& 0& 0& 0& 0& 0& 1& 0&
0& 0& 0& 1& 1& 1& 0\\0& 0& 0& 0& 0& 0& 0& 1& 0& 0& 0& 0& 1& 1& 1\\0&
0& 0 &0& 0& 0& 0& 0& 1& 0& 0& 1& 1& 1& 1\\0& 0& 0& 0& 0& 0& 0& 0& 0&
1& 0& 1& 0& 1& 1\\0& 0& 0& 0& 0& 0& 0& 0& 0& 0&1 &1 &0 &0& 1
\end{array}\right]
\end{eqnarray}

and the code $C_1^\perp$ has parameters  $[15, 4, 8]_2$ and
generator matrix
\begin{eqnarray} \left[\begin{array}{p{0.1cm}p{0.1cm}p{0.1cm}cc ccccc cccccc}
1 &0& 0 &0 &1& 0 &0& 1& 1& 0& 1& 0& 1 &1 &1\\0& 1& 0& 0& 1 &1 &0& 1&
0& 1& 1& 1& 1& 0 &0\\0& 0& 1&
0& 0 &1& 1& 0& 1 &0 &1 &1& 1 &1 &0\\0 &0& 0 &1 &0 &0 &1& 1 &0 &1& 0& 1& 1& 1& 1\\
\end{array}\right]
\end{eqnarray}

Consider a BCH code $C_2$ with parameters $[15,7,5]_2$ that has
designed distance $5$ and generator matrix given by

\medskip

\begin{eqnarray}\left[\begin{array}{p{0.1cm}p{0.1cm}p{0.1cm}cc ccccc cccccc}
   1 &0& 0& 0& 0& 0& 0& 1& 0& 0& 0& 1& 0 &1 &1\\
  0 &1 &0& 0 &0 &0& 0& 1& 1 &0& 0& 1 &1& 1 &0\\ 0& 0& 1& 0& 0& 0& 0& 0& 1& 1& 0& 0& 1& 1& 1\\
0 &0& 0 &1 &0 &0 &0 &1 &0& 1& 1& 1& 0& 0& 0\\ 0& 0& 0& 0& 1& 0& 0& 0& 1& 0& 1& 1& 1& 0& 0\\
0& 0& 0& 0& 0& 1& 0& 0& 0& 1& 0& 1& 1 &1 &0\\ 0& 0& 0 &0 &0& 0& 1& 0
&0 &0 &1 &0& 1& 1& 1
\end{array}\right]
\end{eqnarray}

and the code $C_2^\perp$ has parameters $[15, 8, 4]_2$ and generator
matrix
\begin{eqnarray} \left[\begin{array}{p{0.1cm}p{0.1cm}p{0.1cm}cc ccccc cccccc}
1 &0 &0& 0& 0& 0& 0& 0& 1& 1& 0& 1& 0& 0& 0\\0& 1 &0 &0 &0 &0 &0 &0&
0& 1& 1& 0& 1& 0& 0\\0& 0& 1& 0& 0& 0& 0& 0& 0& 0& 1& 1& 0& 1& 0\\0&
0& 0& 1& 0& 0& 0& 0& 0& 0& 0& 1& 1& 0& 1\\0& 0& 0& 0& 1& 0& 0& 0& 1&
1& 0& 1& 1& 1& 0\\0& 0& 0& 0& 0& 1& 0& 0& 0& 1& 1& 0& 1& 1& 1\\0& 0&
0& 0& 0& 0& 1& 0& 1& 1&1 &0 &0& 1& 1\\0& 0& 0& 0& 0& 0& 0& 1& 1& 0&
1& 0& 0& 0& 1
\end{array}\right]
\end{eqnarray}

\bigskip

\noindent \textbf{AQEC.} We can consider the code $C_1$ corrects the
bit-flip errors such that $C_2^\perp \subset C_1$. Furthermore,
$C_1^\perp \subset C_2$. Furthermore and $d_x=\wt (C_1 \backslash
C_2^\perp)=3$ and $d_z=\wt (C_2 \backslash C_1^\perp)=5$. Hence, the
quantum code  can detect four phase-shift errors and two bit-flip
errors, in other words, the code can correct two phase-shift errors
and one bit-flip errors. There must exist asymmetric quantum error
control codes (AQEC) with parameters $[[n,k_1+k_2-n,d_z/d_x ]]_2
=[[15,3,5/3]]_2$. We ensure that this quantum code encodes three
qubits into $15$ qubits, and  it might also be easy to design a
fault tolerant circuit for this code similar to $[[9,1,3]]_2$ or
$[[7,1,3]]_2$, but one can use the cyclotomic structure of this
code. We ensure that many other quantum BCH can be constructed using
the approach given in this paper that may or may not have better
fault tolerant operations and better threshold values.

\noindent \textbf{SSC.} We can also construct a subsystem code (SSC)
based on the codes $C_1$ and $C_2$. First we notice that $C_1^\perp
=C_2 \cap C_2^\perp \neq \emptyset$,  $C_2 \subset C_1$ and
$C_2^\perp \subset C_1$. Let $k=\dim C_1 - \dim C_2=4$ and $r=\dim
C_2- \dim C_1^\perp = 3$. Furthermore $d=\wt(C_1 \backslash C_2) =
3$. Therefore, there exists a subsystem code with parameters
$[[15,4,3,3]]_2$ also an ASSC code with parameters
$[[15,4,3,5/3]]_2$.

\begin{remark}
An $[7,3,4]_2$  BCH code is used to derive Steane's code
$[[7,1,4/3]]_2$.  AQEC might not be interesting for Steane's code
because it can only detect $3$ shift-errors and  $2$ bit-flip
errors, furthermore, the code corrects one bit-flip and one
phase-shift at most. Therefore, one needs to design AQEC with $d_z$
much larger than $d_x$.

One might argue on how to choose the distances  $d_z$ and $d_x$, we
think the answer comes from the physical system point of view. The
time needed to phase-shift errors is much less that the time needed
for qubit-flip errors, hence depending on the factor between them,
one can design AQEC with factor a $d_z/d_x$.
\end{remark}

\section{Bounds on Asymmetric QEC and  Subsystem
Codes}\label{sec:bounds} One might wonder whether the known bounds
on QEC and SSC parameters would also apply for AQEC and ASSC code
parameters. We can show that AQECs and ASSCs obey the asymmetric
Singleton bound as follows. In fact we can trade the dimensions of
SCC and ASSC in a similar manner as shown in~\cite{aly08a,aly08f}.

\subsection{Singleton Bound}[Asymmetric Singleton Bound]
\begin{theorem}\label{th:singleton}
An  $[[n,k,d_z/d_x]]_q$ asymmetric pure quantum code with $k \geq 1$
satisfies $d_x\leq  (n-k+2)/2$, and the bound
\begin{eqnarray} d_x+d_z\leq  (n-k+2).\end{eqnarray}
\end{theorem}
\begin{proof}
From the construction of AQEC, existence of the AQEC with parameters
$[[n,k,d_z/d_x]]_q$ implies existence of two codes $C_1$ and $C_2$
such that $C_2^\perp \subseteq C_1$ and $C_1^\perp \subseteq C_2$.
furthermore $d_x=\wt(C_1 \backslash C_2^\perp)$ and $d_z=\wt(C_2
\backslash C_1^\perp)$. Hence we have $d_x \leq (n-k_1+1)$ and $d_z
\leq (n-k_2+1)$, and by adding these two terms we obtain $d_x+d_z
\leq n-(k_1+k_2-n)+2=n-k+2$.
\end{proof}
It is much easy to show that the bound for $d_x$ than the bound for
$d_z$ since QEC's with parameters $[[n,k,d_x]]_q$ obey this bound.
Also, impure AQECs obey this bound $d_x+d_z\leq  (n-k+2)$. The proof
is straight forward to the case QECs and we omit it here.

One can also show that Asymmetric subsystem codes obey the Singleton
bound
\begin{lemma}\label{th:boundlind2}
Asymmetric subsystem codes with parameters  $[[n,k,r,d_z/d_x]]_q$
for $0\leq r <k$ satisfy

\begin{eqnarray}
k+r\leq n-d_x-d_z+2.
\end{eqnarray}
\end{lemma}

\begin{remark}
In fact, the AQEC RS codes derived in Section~\ref{sec:AQEC-BCH} are
optimal and asymmetric MDS codes in a sense that they meet
asymmetric Singleton bound with equality. The conclusion is that MDS
QECs are also MDS AQEC. Furthermore, MDS SCC are also MDS ASSC.

\end{remark}

\subsection{Hamming Bound}
Based on the discussion presented in the previous sections, we can
treat subsystem code constructions as a special class of asymmetric
quantum codes where $C_i^\perp \subset C_{1+(i\mod 2)}$, for $i \in
\{1,2\}$ and $C_2=C_{1 } \cap C_{1 }^\perp$. Furthermore, the more
general theory of quantum error control codes would be asymmetric
subsystem codes.

\begin{lemma}
A pure $((n,K,K',d_z/d_x))_q$ asymmetric subsystem code satisfies
\begin{eqnarray}
\sum_{j=0}^{\lfloor\frac{d_x-1}{2} \rfloor}\binom{n}{j}(q^2-1)^j
\leq q^n/KK'.\end{eqnarray}
\end{lemma}
\begin{proof}
We know that a pure $((n,K,K',d_z/d_x)))_q$ code implies the
existence of a pure $((n,KK',d_x))_q$ stabilizer code this is direct
by looking at an AQEC as a QEC. But this obeys the quantum Hamming
bound~\cite{feng04,aly06c}. Therefore it follows that
$$ \sum_{j=0}^{\lfloor\frac{d_x-1}{2} \rfloor}\binom{n}{j}(q^2-1)^j \leq q^n/KK'.$$
\end{proof}
 In terms of packing codes, it is easy to show that the impure
asymmetric subsystem codes does not obey the quantum Hamming bound.
Since the special case does not obey this bound, so why the general
case does.
\begin{lemma}
An impure $((n,K,K',d_z/d_x))_q$ asymmetric subsystem code does not
satisfy
$$ \sum_{j=0}^{\lfloor\frac{d_x-1}{2} \rfloor}\binom{n}{j}(q^2-1)^j \leq q^n/KK'.$$
\end{lemma}

It is obvious that the distance of phase-shift would not obey this
bound as well, $d_z > d_x$. Finally one can always look at
asymmetric quantum codes (AQECs) as a special class of asymmetric
subsystem codes (ASSCs). In other words every an $[[n,k,d_z/d_x]]_q$
is also an $[[n,k,0,d_z/d_x]]_q$, and this is the main contribution
of this paper. Also, a SSC with parameters $[[n,k,r,d_x]]_q$ can
produce $ASSC$ with parameters $[[n,k,r,d_z/d_x]]_q$. One can also
go from ASSCs to AQECs using the results derived
in~\cite{aly08a,aly08f}. and Finally an ASSC with parameters
$[[n,k,r,d_z/d_x]]_q$ is also an ASSC with parameters
$[[n,r,k,d_z/d_x]]_q$. The proof for all these facts is a direct
consequence by writing the $\F_q$ bases for the codes AQEC and ASSC.

\section{Conclusion and Discussion}\label{sec:conclusion}
This paper introduced a new theory of asymmetric quantum codes.  It
establishes a link between asymmetric and symmetric quantum control
codes, as well as subsystem codes. Families of AQEC are derived
based on RS and BCH codes over finite fields. Furthermore we
introduced families of subsystem BCH codes. Tables of AQEC-BCH and
CSS-BCH are shown over $\F_q$.

We pose it as open quantum to study the fault tolerance operations
of the constructed quantum BCH codes in this paper.  Some BCH codes
are turned out to be also LDPC codes. Therefore, one can use the
same method shown in to construct asymmetric quantum LDPC
codes~\cite{aly08c}.

\section*{Acknowledgments.}
I  thank A. Klappenecker for his support and  I think my family,
teachers, and colleagues.

Part of this research on SSC and QEC has been done at CS/TAMU in
Spring '07 and during a research visit to Bell-Labs \&
alcatel-Lucent in Summer '07, the generalization to ASSC is a
consequence.

\bigbreak
\noindent
\begin{eqnarray*}\mathfrak{Sharing~knowledge,~in~which~we~all~born~
knowing~nothing,}\\
\mathfrak{~is~better~than~proving~or~canceling~it.~~S.A.A.}
\end{eqnarray*}

\section{Appendix}

\subsection{Quantum BCH Codes}

This paper is written on the occasion of the 50th anniversary of the
discovery  of classical BCH codes and their quantum counterparts
were  derived nearly 10 years ago. This powerful class of codes has
been used for the construction of  quantum block and convolutional
codes, entangled-assisted quantum convolutional codes, and subsystem
codes; in addition to the constructions  of classes of low-density
parity check (LDPC)
codes~\cite{aly06a,aly07a,aly08thesis,aly08b,steane96,steane99,brun06,aly07d,aly08a}.

\scriptsize
\newcommand{\XXstud}{{}}
\newcommand{\XXar}[1]{}

\end{document}